\documentclass[a4paper,11pt,twoside]{article}

\usepackage{graphicx}
\usepackage[T1]{fontenc} 
\usepackage[utf8]{inputenc}
\usepackage[english]{babel}
\usepackage{soul}
\usepackage{amsfonts}
\usepackage{amsmath}
\usepackage{amsthm}
\usepackage{amssymb}
\usepackage{mathrsfs}
\usepackage[top=4cm, bottom=3cm, left=3.4cm, right=3.4cm]{geometry}
\usepackage{setspace}
\usepackage{afterpage}
\usepackage{extarrows}
\usepackage{fancyhdr}
\usepackage{titlesec}
\usepackage{enumitem} \setlist{nosep}
\usepackage[pdftex,breaklinks,colorlinks,linkcolor=blue,
anchorcolor=blue]{hyperref}

\theoremstyle{definition}
\newtheorem{defin}{Definition}[section]
\theoremstyle{plain}
\newtheorem{theo}[defin]{Theorem}
\newtheorem{lem}[defin]{Lemma}

\newtheorem{cor}[defin]{Corollary}
\theoremstyle{definition}
\newtheorem{exm}[defin]{Example}

\newtheorem{rems}[defin]{Remarks}

\newcommand{\mc}{\mathcal}
\newcommand{\D}{\mathcal{D}}

\renewcommand{\H}{\mathcal{H}}

\newcommand{\Sc}{\mathcal{S}}
\newcommand{\B}{\mathcal{B}}
\newcommand{\F}{\mathcal{F}}
\newcommand{\n}[1]{\|#1\|}

\renewcommand{\l}{\langle}
\renewcommand{\r}{\rangle}
\newcommand{\N}{\mathbb{N}}

\newcommand{\R}{\mathbb{R}}

\newcommand{\pin}[2]{\l#1 , #2\r}

\newcommand{\Hil}{\mc H}
\newcommand{\1}{1 \!\! 1}
\newcommand{\Lc}{{\cal L}}
\newcommand{\ltwo}{{\Lc^2(\mathbb{R})}}

\newcommand{\scr}{\Sc(\mathbb{R})}
\newcommand{\ldue}{\Lc^2(\mathbb{R})}

\newcommand{\be}{\begin{equation}}
\newcommand{\en}{\end{equation}}
\newcommand{\bea}{\begin{eqnarray}}
\newcommand{\ena}{\end{eqnarray}}
\newcommand{\beano}{\begin{eqnarray*}}
	\newcommand{\enano}{\end{eqnarray*}}
\newcommand{\bee}{\begin{enumerate}}
	\newcommand{\ene}{\end{enumerate}}

\numberwithin{equation}{section}

\titleformat{\section}
{\normalfont\fillast \fontsize{12}{15}\scshape}{\thesection.}{0.8em}{}

\titleformat{\subsection}
{\normalfont\fillast \fontsize{11}{12}\scshape}{\thesubsection.}{0.8em}{}

\begin{document}
	
	\thispagestyle{empty}

	\vspace*{2cm}
	
\begin{center}
	{\Large \bf Some perturbation results for quasi-bases and other sequences of vectors}   \vspace{2cm}\\
	
	{\large Fabio Bagarello}\\
	Dipartimento di  Ingegneria, \\
	Universit\`a di Palermo,\\ I-90128  Palermo, Italy, and\\
	INFN, Sezione di Catania, Italy\\
	e-mail: fabio.bagarello@unipa.it\\
	
	\vspace{4mm}
	
	{\large Rosario Corso}\\
	Dipartimento di Matematica ed Informatica,\\
	Universit\`a di Palermo,\\ I-90123 Palermo, Italy\\
	e-mail: rosario.corso02@unipa.it

\end{center}
	
	\vspace*{0.2cm}
	
	\begin{abstract}
		\noindent 	We discuss some perturbation results concerning certain pairs of sequences of vectors in a Hilbert space $\Hil$ and producing new sequences which share, with the original ones,  { reconstruction formulas on a dense subspace of $\Hil$ or on the whole space}.  We also propose some preliminary results on the same issue, but in a distributional settings.		
	\end{abstract}
	
	\vspace{2cm}
	
	
	\vfill

	%\pagenumbering{roman}

	\newpage

\section{Introduction and notation}

In this paper we consider different expansions of elements of a Hilbert space, or of some space of distributions, in terms of a sequence in the style of Schauder bases. In particular, we propose some perturbation results for these expansions. Before going into details, let us fix our framework and our notation.   Let $\H$ be a Hilbert space with inner product $\pin{\cdot}{\cdot}$ (which we consider linear in the second component and anti-linear in the first component) and norm $\|\cdot\|$. 

Throughout the paper $\D$ is a dense subspace of $\H$, $\B(\H)$ is the set of bounded operators $T:\H\to \H$. We denote by $T^\dagger$ the adjoint of the operator $T\in \B(\H)$.

We recall that a Schauder basis $\F_\xi=\{\xi_n\}$ for $\H$ is a sequence of vectors $\xi_n$ satisfying 
$$
f=\sum_{n=0}^\infty  \pin{\chi_n}{f}\xi_n, \qquad \forall f\in \H,
$$
with a unique sequence $\F_\chi=\{\chi_n\}$. A special class of Schauder bases for Hilbert spaces is constituted by Riesz bases. A Riesz basis  $\F_\xi=\{\xi_n\}$ is a sequence of elements of $\H$ such that $\xi_n=Te_n$ for every $n\in \N$, where $\F_e=\{e_n\}$ is an orthonormal (o.n.) basis and $T:\H\to \H$ is a bounded and bijective operator. Moreover $\pin{\xi_n}{\widetilde \xi_m}=\delta_{n,m}$ for every $n,m$, where $\F_{\widetilde \xi}=\{\widetilde\xi_n\}$, $\widetilde \xi_n=(T^{-1})^\dagger e_n$ is a Riesz basis, too. 
Equivalently, $\F_\xi$ is a Riesz basis if it is a complete\footnote{A sequence $\F_\xi$ of elements of $\H$ is complete if its linear span is dense in $\H$.} sequence and there exist $A,B>0$ such that 
$$
A \sum_{n=0}^\infty |c_n|^2 \leq \Big \|\sum_{n=0}^\infty c_n \xi_n\Big\|^2 \leq B \sum_{n=0}^\infty |c_n|^2 
$$
for every finite complex sequence $\{c_n\}$\footnote{A complex sequence $\{c_n\}$ is {\it finite} is $c_n\neq 0$ for only a finite number of elements.}.

Frames in Hilbert spaces, introduced by Duffin and Schaffer \cite{dufsch}, are generalizations of Riesz bases and are defined as follows.  A frame $\F_\xi=\{\xi_n\}$ for $\H$ is a sequence of elements of $\H$ such that
$$
A \n{f}^2 \leq \sum_{n=0}^\infty |\pin{\xi_n}{f}|^2 \leq B \n{f}^2, \qquad \forall f\in \H
$$
for some $A,B>0$. A frame $\F_\xi=\{\xi_n\}$ induces a reconstruction formula (strong version of the resolution of the identity) of the type
\begin{equation}
\label{rec_frame}
f=\sum_{n=0}^\infty  \pin{\xi_n}{f}\chi_n=\sum_{n=0}^\infty  \pin{\chi_n}{f}\xi_n, \qquad \forall f\in \H,
\end{equation}
where  $\F_\chi=\{\chi_n\}$ is some other frame for $\H$. Equation \eqref{rec_frame} generalizes the expansion given above for Schauder basis. For instance, the vectors in $\F_\xi$ need not being linearly independent.

Later we will also need the notion of Bessel sequence, which is a sequence of vectors such that only one of the inequalities above for frame holds true. More explicitly, $\F_\xi=\{\xi_n\}$ is a Bessel sequence if $$
\sum_{n=0}^\infty |\pin{\xi_n}{f}|^2 \leq B \n{f}^2, \qquad \forall f\in \H
$$
for some $B>0$.

In this paper we are particularly interested in perturbation results of sequences of vectors. 
For instance, for Schauder bases, Riesz bases and frames the perturbation results in the following lemmas can be found in the literature. The idea of these results is essentially that if a sequence $\F_\xi$ has some property (being a frame, a Riesz basis or a Schauder basis), then a sequence $\F_\varphi$ which is {\em close to} $\F_\xi$ (in a sense we will properly define) has the same property. 

In what follows the following perturbation result for operators will be used heavily.

\begin{lem}[{{\cite[Theorem IV.1.16]{Kato}}}]
	\label{lem_pert_op}
	Let $S,T\in \B(\H)$ be such that $T^{-1}\in \B(\H)$. If $\n{S}< \n{T^{-1}}^{-1}$, then $T+S$ is invertible with $(T+S)^{-1}\in \B(\H)$.  
\end{lem}

In particular, some (well known) consequences of this Lemma are {the following perturbation results for Schauder bases, Riesz bases and frames we mentioned above. }

\begin{lem}[{\cite[Corollary 22.1.5]{Chris}}]
	\label{pert2}
	Let $\F_\xi$ be a frame for $\H$ with bounds $A,B$. Let $\F_\varphi$ be a sequence in $\H$ and assume that there exists $B'<A$ such that
	\begin{equation}\label{11}
	\sum_{n=0}^\infty |\pin{f}{\xi_n-\varphi_n}|^2\leq B' \n{f}^2, \qquad \forall f\in \H. 
	\end{equation}
	Then $\F_\varphi$ is a frame for $\H$ with bounds
	$$
	A\left (1- \sqrt{\frac{B'}{A}}\right )^2, \qquad B\left (1+ \sqrt{\frac{B'}{B}}\right )^2. 
	$$  
	Moreover, if $\F_\xi$ is a Riesz basis for $\H$, then also $\F_\varphi$ is a Riesz basis for $\H$. 
\end{lem}

\begin{lem}\label{pertcn}
	Let $\F_\xi$ and $\F_\varphi$ be two sequences of $\H$, and assume that there exists $0 \leq \alpha <1$ such that 
	\begin{equation}
	\Bigg \|\sum_{n=0}^\infty c_n (\xi_n-\varphi_n)\Bigg \|\leq \alpha \Bigg \| \sum_{n=0}^\infty c_n \xi_n\Bigg \|. 
	\end{equation}
	for every finite complex sequence $\{c_n\}$.
	\begin{enumerate}
		\item If  $\F_\xi$ be a Schauder basis for $\H$, then $\F_\varphi$ is a Schauder basis for $\H$. 
		\item  If $\F_\xi$ is a frame for $\H$ with bounds $A,B$, then $\F_\varphi$ is a frame for $\H$ with bounds
		\begin{equation}
		\label{bounds_pert}
		A (1-  \alpha  )^2, \qquad B (1+  \alpha   )^2. 
		\end{equation}
		\item If $\F_\xi$ is a Riesz basis for $\H$, then also $\F_\varphi$ is a Riesz basis for $\H$ with bounds given by \eqref{bounds_pert}.
	\end{enumerate}
\end{lem}

Lemma \ref{pertcn} point 1 is a result of Paley and Wiener \cite{paleywiener} and actually holds also for Schauder bases for Banach spaces \cite[Ch.1-Th. 10]{young}. Lemma \ref{pertcn} points 2 and 3 are proved in \cite[Theorem 22.1.1]{Chris}. 

During the last decades other types of sequences of vectors have been defined and studied, like $\D$-quasi bases, originally introduced in \cite{bag2013}, and then studied and extended in \cite{bagbookPT,bagspringer,inoue1,inoue2,kamuda}. Let $\D$ be a dense subspace of $\H$. A pair of sequences $(\F_\xi,\F_{\widetilde \xi})$ is a $\D$-quasi basis if $\pin{\xi_n}{\widetilde\xi_m}=\delta_{n,m}$ and 
\be
\sum_{n=0}^\infty \pin{f}{\xi_n}\pin{\widetilde\xi_n}{g}=\pin{f}{g}, \qquad  \forall f,g\in \D.
\label{dquasi}\en
$\D$-quasi bases seem to be valid alternatives to ordinary bases in concrete physical applications. In particular, as (\ref{dquasi}) shows, they produce resolutions of the identity (or, as they are called with a more physical language, {\em closure relations}), for some physical systems described by certain Hamiltonian operators which are not self-adjoint and whose eigenstates are known not to produce basis in $\Hil$. This aspect is discussed in details in, e.g., \cite{bagspringer}. 

Inspired by Lemma \ref{pert2} and \ref{pertcn}, in this paper we will provide perturbation results for sequences (in particular, for $\D$-quasi bases) generating a reconstruction formula on a dense subspace $\D$ of $\H$, even when the original sequences do not generate themselves any reconstruction formula. Since we are dealing with sequences $\F_\xi$ for which a series $\sum_{n=0}^\infty c_n \xi_n$ may not be convergent, we will consider also convergence and reconstruction formulas in a weak sense. Throughout the paper we will discuss some examples also in connection to Physics, and to Quantum Mechanics in particular.

The paper is organized as follows:  in Section \ref{sec2} we give perturbation results regarding reconstruction formulas in strong sense. The weak versions of these results are discussed in Section \ref{sec3}, and extended further to a distributional settings in Section \ref{sec4}. Section \ref{sec6} contains our conclusions.   

\vspace*{2mm}

{\bf Remark.} 
	As we will show, quite often in our results a parameter $\lambda$ appears to extend the range of their applications. To make an example with Lemma \ref{pert2}, given a frame $\F_\xi$ with bounds $A,B$, a sequence  $\F_\varphi$ may not satisfy \eqref{11}, but it may still satisfy, for some $\lambda \neq 0$ and $B'<A$, the condition
	\begin{equation}
		\sum_{n=0}^\infty |\pin{f}{\xi_n-\lambda\varphi_n}|^2\leq B' \n{f}^2, \qquad \forall f\in \H.
	\end{equation}
	Then, the conclusion of that lemma still holds, i.e. $\F_\varphi$ is a frame for $\H$.

\section{Convergence in strong sense}
\label{sec2}

We begin our series of perturbation results considering first what can be deduced when working in strong sense.

\begin{theo}\label{theom41}
	Let $\D$ be a dense subspace of $\H$. Let $\F_\varphi$ and $\F_\psi$ be sequences of $\H$ such that $\displaystyle \sum_{n=0}^\infty \pin{\varphi_n}{f}\psi_n$ exists for every $f\in \D$ and let $0\leq\alpha<1$ and $\lambda \neq 0$ be such that 
	\begin{equation}
	\label{pertDbasis}
	\left \|\lambda\sum_{n=0}^\infty \pin{\varphi_n}{f}\psi_n- f \right\|\leq \alpha \|f\|, \qquad \forall f\in \D. 
	\end{equation}
	Then, a sequence $\F_{\widetilde{\varphi}}$ does exist such that
	\be
	f=\sum_{n=0}^\infty \pin{\varphi_n}{f}\widetilde\varphi_n, \qquad \forall f\in \D. 
	\label{13}\en
	If, in particular, $\D=\H$, a second sequence $\F_{\widetilde{\psi}}$ also exists such that
	\be
	f=\sum_{n=0}^\infty \pin{\varphi_n}{f}\widetilde\varphi_n=\sum_{n=0}^\infty \pin{\widetilde\psi_n}{f}\psi_n, \qquad \forall f\in \Hil,
	\label{14}\en
	satisfying   $\pin{\varphi_n}{\widetilde{\varphi}_m}=\pin{\widetilde\psi_n}{{\psi}_m}$  for every $n,m$. In particular, $\F_\varphi$ and $\F_{\widetilde \varphi}$ are bi-orthogonal if and only if $\F_{\psi}$ and $\F_{\widetilde \psi}$ are bi-orthogonal.
	% $\pin{\varphi_n}{\widetilde{\varphi_m}}=\delta_{n,m}$ if and only if $\pin{\psi_n}{\widetilde{\psi_m}}=\delta_{n,m}$.
\end{theo}
\begin{proof}
	Let us define $Q:\D\to \H$, $\displaystyle Qf=\sum_{n=0}^\infty \pin{\varphi_n}{f}\psi_n$ for $f\in \D$, which is a bounded operator by \eqref{pertDbasis}. Since $\D$ is dense in $\H$, $Q$ has a bounded extension $\overline{Q}:\H\to \H$ which satisfies 
	\begin{equation*}
	\left \|\lambda\overline{Q}f- f \right\|\leq \alpha \|f\|, \qquad \forall f\in \H. 
	\end{equation*}
	Note that $\overline{Q}=Q$ if $\D=\H$. 	By Lemma \ref{lem_pert_op}, putting  $S=\lambda\overline{Q}- I$ and $T=I$, we obtain that $\overline{Q}$ is invertible with $\overline{Q}^{-1}\in \B(\H)$. As a consequence, $Q$ is injective and $\overline{Q}^{-1}Q=I_{\D}$ and $Q\overline{Q}^{-1}=I_{R(Q)}$. 
	If we define $\widetilde \varphi_n=\overline{Q}^{-1}\psi_n$ and $\widetilde \psi_n={\overline{Q}^{-1}}^\dagger\varphi_n$, $n\geq 0$, we have
	$$
	\sum_{n=0}^\infty \pin{\varphi_n}{f}\widetilde\varphi_n=\overline{Q}^{-1}\sum_{n=0}^\infty \pin{\varphi_n}{f}\psi_n=\overline{Q}^{-1}Qf=f, \qquad\forall f\in \D,
	$$
	$$
	\sum_{n=0}^\infty \pin{\widetilde \psi_n}{f}\psi_n=\sum_{n=0}^\infty \pin{\varphi_n}{\overline{Q}^{-1} f}\psi_n=Q\overline{Q}^{-1}f=f, \qquad\forall f\in R(Q).
	$$
	Here we have used, in particular, the boundedness of $\overline{Q}^{-1}$. When $\D=\H$, $R(Q)=\H$.  Finally, the last claim follows from the identities $\pin{\varphi_n}{\widetilde{\varphi}_m}=\pin{\varphi_n}{\overline{Q}^{-1}{\psi_m}}$, together with $\pin{\psi_n}{\widetilde{\psi}_m}=\pin{\overline{Q}^{-1}\psi_n}{{\varphi}_m}$.
	
\end{proof}

\begin{exm}
	\label{exl2R1} 
	Let $\H=\mathcal{L}^2(\R)$ and $\rho_j$, $j=1,2$, two Lebesgue-measurable functions. Let us introduce the functions
	\be
	\varphi_n(x)=\rho_1(x)\,c_n(x),\qquad \psi_n(x)=\rho_2(x)\,c_n(x),
	\label{f2.1}\en
	where $c=\{c_n\}$ is taken to be an o.n. basis of $\ltwo$. We write, as usually,  $\F_\varphi=\{\varphi_n\}$ and $\F_\psi=\{\psi_n\}$. The properties of these two sequences have been discussed, in connection with some physical systems, in \cite{bagspringer} and references therein, under the assumption that
	\be \label{f3}
	\rho_1(x)\overline{\rho_2(x)}=1
	\en
	a.e. in $\mathbb{R}$. In this example, we assume that $\rho_1,\rho_2$ have real values, $\rho_2\in \mathcal{L}^\infty(\R)$, $\|\rho_1 \rho_2\|_\infty<\infty $ and $ \rho_1(x)\rho_2(x)>c>0$ in $\R$. Furthermore, we choose $\rho_1$ so that $\varphi_n\in \Lc^2(\R)$ for every $n$. It is clear that  $\psi_n\in \Lc^2(\R)$ for every $n$, due to the boundedness of $\rho_2$. \\
	We define $\D=\{f\in \mathcal{L}^2(\R): \rho_1 f\in \mathcal{L}^2(\R)\}$. This subspace is dense in $\H$ independently of $\rho_1$ (see \cite[Example 3.8]{schm}). 
In other words, $\D$ is a {\em large} set, for any choice of $\rho_1$.
	
	For $f\in \D$,  $\sum_{n=0}^\infty \pin{\varphi_n}{f}\psi_n$ is convergent: in fact,  $\{\pin{\varphi_n}{f}\}\in l^2(\N)$ and  $\F_\psi=\{\psi_n\}$ is a Bessel sequence. 
	Moreover, because of (\ref{f2.1}),
	$$
	Q f= \sum_{n=0}^\infty \pin{\varphi_n}{f}\psi_n=\rho_1 \rho_2 f, \qquad \forall f\in \D.
	$$
	Therefore, if we set $0<\lambda < \frac{2}{\n{\rho_1\rho_2}_\infty}$
	we conclude that
	%\be\label{r1r2_bound}
	$$
	\left  \|\lambda \sum_{n=0}^\infty \pin{\varphi_n}{f}\psi_n- f \right \|=\left\|(\lambda\rho_1 \rho_2 -1)f\right\|\leq \alpha\|f\| \qquad \forall f\in \D,
	$$%\en
	where $\alpha=\|\lambda\rho_1 \rho_2-1 \|_\infty$. By our hypothesis, $\alpha < 1$, then by Theorem \ref{theom41} a sequence $\F_{\widetilde{\varphi}}$ does exist such that \eqref{13} holds. Following the proof of the theorem, we can also give an explicit expression of $\F_{\widetilde{\varphi}}$. Indeed, clearly, $\overline{Q}:\mathcal{L}^2(\R) \to \mathcal{L}^2(\R)$ is the multiplication operator by $\rho_1 \rho_2$ on $\mathcal{L}^2(\R)$.  $\widetilde \varphi_n=\overline{Q}^{-1}\psi_n=\rho_1^{-1}c_n$ for any $n$. We want to remark that,  by our hypothesis, $\rho_1^{-1}\in \mathcal{L}^\infty(\R)$. \\
	If in addition also $\rho_1\in\mathcal{L}^\infty(\R)$ (which, together the assumptions above, implies that $\rho_1^{-1},\rho_2^{-1}\in\mathcal{L}^\infty(\R) $ so $\F_\varphi$ and $\F_\psi$ are Riesz bases), then \eqref{14} holds and, in particular, $\widetilde \varphi_n(x)=\rho_1^{-1}(x)c_n(x)$, $\widetilde \psi_n(x)=\rho_2^{-1}(x)c_n(x)$ for any $n$. It is also clear that
	$$
	\pin{\varphi_n}{\widetilde{\varphi}_m}=\pin{\widetilde\psi_n}{{\psi}_m}=\pin{c_n}{c_m}=\delta_{n,m},
	$$
	independently of the explicit choice of $\rho_j$.
	\end{exm}

{
\begin{exm}\label{exa1}
	Let $\F_\eta$	and $\F_{\chi}$ be sequences such that
	$$
	f=\sum_{n=0}^\infty \pin{f}{\eta_n}\chi_n, \qquad \forall f\in\Hil.
	$$
	Let $X,Y\in \B(\Hil)$ be two non necessarily invertible operators with $\|X\|+\|Y\|+\|X\|\|Y\|<1$. Then, if we perturb $\F_{\eta}$ and $\F_{\chi}$ as follows, $\F_\varphi=\{\varphi_n=(\1+X)\eta_n\}$ and $\F_\psi=\{\psi_n=(\1+Y)\chi_n\}$, these two sequences satisfy (\ref{pertDbasis}) for $\lambda=1$. This can be checked using the following equality
	$$
	\sum_{n=0}^\infty \pin{f}{\varphi_n}\psi_n-f={X^\dagger f}+{Yf}+{YX^\dagger f},
	$$
	together with the Schwarz inequality and our condition on $\|X\|$ and $\|Y\|$. A particular choice of $X$ and $Y$, manifestly non invertible, can be produced by any fixed vector $\sigma\in\Hil$, normalized: $\|\sigma\|=1$. We define the operators $P_\sigma f=\pin{\sigma}{f}\sigma$, and $R_\sigma=\1-P_\sigma$. It is well known that they are both orthogonal projectors and, as such, in particular $\|P_\sigma\|=\|R_\sigma\|=1$. Also, $P_\sigma R_\sigma=0$. Then, if we take for instance $X=\alpha\,P_\sigma$ and $Y=\beta\,R_\sigma$, $\alpha,\beta>0$ with $\alpha+\beta+\alpha\beta<1$, we can perturb the original sequences $\F_{\eta}$ and $\F_{\chi}$ as proposed above. In particular, in this example we can identify all the vectors and the operators. Indeed we find
	$$
	Q=\1(1+\beta)+(\alpha-\beta)P_\sigma, \qquad Q^{-1}=\frac{1}{1+\beta}\left(\1+\frac{\beta-\alpha}{1+\alpha}\,P_\sigma\right).
	$$
	We observe that $Q=Q^\dagger$ and $Q^{-1}=(Q^{-1})^\dagger$. As for the vectors, we find
	$$
	\varphi_n=(\1+\alpha P_\sigma)\eta_n=\eta_n+\alpha P_\sigma\eta_n, \qquad \psi_n=(\1+\beta R_\sigma)\chi_n=\chi_n+\beta R_\sigma\chi_n,
	$$
	together with
	$$
	\widetilde{\varphi}_n=Q^{-1}\psi_n=(\1+\alpha P_\sigma)^{-1}\chi_n=\left(\1-\frac{\alpha}{1+\alpha}\,P_\sigma\right)\chi_n
	$$ 
	and
	$$
	\widetilde{\psi}_n=(Q^{-1})^\dagger\varphi_n=(\1+\beta R_\sigma)^{-1}\eta_n=\left(\1-\frac{\beta}{1+\beta}\,R_\sigma\right)\eta_n.
	$$
\end{exm}}

\begin{rems}
	\label{rem2.8}
	\begin{enumerate}
		\item Note that under the assumptions of Theorem \ref{theom41} and in the case $\D=\H$ we do not always have that 
		\begin{equation}
		\label{rem3.3}
		f=\sum_{n=0}^\infty\pin{\widetilde\varphi_n}{f}\varphi_n \quad\text{ or } \quad  f=\sum_{n=0}^\infty\pin{\psi_n}{f}\widetilde\psi_n, \quad \forall f\in \H,
		\end{equation} 
		because these series may not be convergent\footnote{In any case, \eqref{rem3.3} are always true in weak sense by \eqref{14}, meaning that $
			\pin{f}{g}=\sum_{n=0}^\infty\pin{\widetilde\varphi_n}{f}\pin{\varphi_n}{g}=\sum_{n=0}^\infty\pin{\psi_n}{f}\pin{\widetilde\psi_n}{g}$,  $ \forall f,g\in \H$.}.  Indeed, let $\{e_n\}$ be an o.n. basis of $\H$ and let  $$\F_\varphi=\{e_1,e_1,e_1,e_2,e_2,e_2,e_3,e_3,e_3,\dots\} \text{ and } \F_{\widetilde\varphi}=\{e_1,e_1,-e_1,e_2,e_1,-e_1,e_3,e_1,-e_1,\dots\},$$
		see \cite[Example 2.2]{balsto}. It is easy to see that $\displaystyle f=\sum_{n=0}^\infty \pin{\varphi_n}{f}\widetilde\varphi_n$ for any $f\in \H$, but $\displaystyle f=\sum_{n=0}^\infty \pin{\widetilde\varphi_n}{f}\varphi_n$ only if $\pin{e_1}{f}=0$.

		{
		If in Theorem \ref{theom41} we add that the series $\displaystyle \sum_{n=0}^\infty\pin{\widetilde\varphi_n}{f}\varphi_n$ is convergent for all $f\in \H$, or $\displaystyle \sum_{n=0}^\infty \pin{\psi_n}{f}\varphi_n$  is convergent for all $f\in \H$,  then we can conclude that $\displaystyle f=\sum_{n=0}^\infty\pin{\widetilde\varphi_n}{f}\varphi_n$ for all $f\in \H$ (a similar statement holds with $\psi_n,\widetilde \psi_n$ instead of $\widetilde\varphi_n, \varphi_n$). } For example, these are the cases when  $\F_\varphi$ and $\F_\psi$ are Bessel sequences.
%		If $\F_\varphi$ (resp., $\F_{\psi}$) is a Bessel sequence, then we have also $\displaystyle f=\sum_{n=0}^\infty\pin{\widetilde\varphi_n}{f}\varphi_n$ (resp., $\displaystyle f=\sum_{n=0}^\infty\pin{\psi_n}{f}\widetilde\psi_n$), $\forall f\in \H$.
		\item In Theorem \ref{theom41}, if $\pin{\varphi_n}{\psi_m}=\delta_{n,m}$, then $\F_{\widetilde \varphi}= \F_\psi$ and $\F_{\widetilde \psi}= \F_\varphi$. Moreover, 
		$\F_{\widetilde\varphi}$ (resp., $\F_{\widetilde \psi}$) is linearly independent (a Bessel sequence, a frame) if and only if  $\F_\psi$ is so (resp., $\F_\varphi$). 
		\item The excess of a sequence $\F_\varphi$ is  $$e(\varphi)=\sup\{|I|:I\subseteq \N \text{ and } \overline{{span}}\{\varphi_n\}_{n\in \N\backslash I}=\overline{{span}}\{\varphi_n\}_{n\in \N}\}.$$
		Here with $\overline{{span}}\{\varphi_n'\}$ we denote the closed linear span of a sequence $\{\varphi_n'\}$. 
		The excess is a measure of overcompleteness of a sequence. For instance, a frame $\F_\varphi$ is a Riesz basis if and only if $e(\varphi)=0$. Let $\F_\varphi,\F_{\widetilde \varphi}, \F_\psi$ and $\F_{\widetilde \psi}$ be as in Theorem \ref{theom41}. Then $e(\widetilde \varphi)=e(\psi)$ and $e(\widetilde \psi)=e(\varphi)$.  If $\F_\varphi$ and $\F_\psi$ are frames, then $e(\widetilde \varphi)=e(\psi)=e(\widetilde \psi)=e(\varphi)$ by \cite[Theorem 2.2]{beric}.
	\end{enumerate}
\end{rems}

\begin{cor}
	\label{cor_str}
	Let $\D$ be a dense subspace of $\H$. Let $\F_\xi$ and $\F_\chi$ be two sequences of vectors in $\Hil$ such that
	$$
	f=\sum_{n=0}^\infty \pin{\xi_n}{f}\chi_n, \qquad \forall f \in \D.
	$$
	Moreover, let $\F_\varphi$ and $\F_\psi$ be two sequences of $\H$ such that, for some $0\leq\beta, \gamma$ and $\lambda\neq 0$, with $\beta+\gamma<1$, 
	$$
	\left \|\sum_{n=0}^{\infty} \pin{\lambda\varphi_n- \xi_n}{f}\psi_n \right\|\leq \beta \n{f}, \qquad \forall f\in \D, 
	$$
	$$
	\left\|\sum_{n=0}^{\infty} \pin{\xi_n}{f}(\psi_n-\chi_n)\right \|\leq \gamma \n{f}, \qquad \forall f\in \D.  
	$$
	Then a sequence $\F_{\widetilde \varphi}$ does exist such that \eqref{13} holds. If, in particular, $\D=\H$, then two sequences $\F_{\widetilde \varphi} $ and $\F_{\widetilde \psi}$ do exist such that \eqref{14} holds.
\end{cor}

\begin{proof}
	The proof is based on the identity
	$$
	\sum_{n=0}^{\infty} \pin{\lambda\varphi_n}{f}\psi_n-f=\sum_{n=0}^{\infty} \pin{\lambda\varphi_n-\xi_n}{f}\psi_n+\sum_{n=0}^{\infty} \pin{\xi_n}{f}(\psi_n-\chi_n)
	$$
	together with the inequalities above and Theorem \ref{theom41}.
	\end{proof}

\begin{exm}
	Let $\F_\xi$ be a lower semi-frame of $\H$. This means that $\F_\xi$ is a sequence satisfying
	$$
	A\n{f}^2 \leq \sum_{n=0}^{\infty} |\pin{\xi_n}{f}|^2, \qquad \forall f\in \H, 
	$$ 
	for some $A>0$. 
	Then by \cite[Proposition 3.4]{casazza} (see also \cite[Remark 1]{corso_seq} and \cite[Theorem 4.1]{corso_seq2}) there exists a Bessel sequence $\F_\chi$ such that 
	$$
	f=\sum_{n=0}^\infty \pin{\xi_n}{f}\chi_n, \qquad \forall f \in \D,
	$$
	where $ \D=\left \{f\in \H:\sum_{n=0}^{\infty} |\pin{\xi_n}{f}|^2 <\infty \right \}$ and a Bessel bound of $\F_\chi$ is $\frac 1 A$. 
	In the following, we suppose that $\D$ is dense in $\H$. Let $\F_\varphi$ be a sequence of $\H$ such that $\F_{\varphi-\xi}$ is a Bessel sequence with a Bessel bound $B<A$. Then 
	$$
	\left \|\sum_{n=0}^{\infty} \pin{\varphi_n-\xi_n}{f}\chi_n \right\|\leq \sqrt{\frac B A} \n{f}, \qquad \forall f\in \D. 
	$$
	By Corollary \ref{cor_str} (setting $\F_\psi=\F_\chi$, $\beta=\sqrt{\frac B A}$, $\gamma=0$ and $\lambda=1$), there exists a sequence $\F_{\widetilde \varphi}$ such that
	$$
	f=\sum_{n=0}^\infty \pin{\varphi_n}{f}\widetilde\varphi_n, \qquad \forall f\in \D. 
	$$
	We stress that \eqref{13} does not necessarily extend to every $f\in \H$, despite the fact that the formula defines an identity. Examples, involving in particular lower semi-frames, describing this case can be found in \cite{casazza,corso_seq2}. 
\end{exm}

\begin{cor}
	\label{cor2.8}
	Let $\D$ be a dense subspace of $\H$. Let $\F_\xi$ and $\F_{\chi}$ be two sequences of vectors in $\Hil$ such that
	$$
	f=\sum_{n=0}^\infty \pin{\xi_n}{f}\chi_n, \qquad \forall f \in \D.
	$$
	Moreover, let $\F_\varphi$ and $\F_{\psi}$ be two sequences of $\H$ and $\lambda\neq 0$ and suppose that one of the following condition is satisfied 
	\begin{enumerate}
		\item $\F_{\lambda\varphi-\xi}$, $\F_{\psi}$, $\F_{\xi}$ and $\F_{\psi-\chi}$ are Bessel sequences with Bessel bounds $B_{\lambda\varphi-\xi}$, $B_{\psi}$, $B_{\xi}$ and $B_{\psi-\chi}$, respectively, verifying $B_{\lambda\varphi-\xi}B_{\psi}+B_{\xi}B_{\psi-\chi}<1$.	
		\item $\F_{\lambda\varphi-\xi}$ and $\F_{\chi}$ are Bessel sequences with Bessel bounds $B_{\lambda\varphi-\xi}$ and $B_{\chi}$, respectively, verifying $B_{\lambda\varphi-\xi}B_{\chi}<1$.	
		\item 
		\begin{equation}
		\label{cond_cor3.8}
		\sum_{n=0}^{\infty} (\n{\lambda\varphi_n-\xi_n}\n{\psi_n} +\n{\xi_n}\n{\psi_n-\chi_n}) < 1. 
		\end{equation}
	\end{enumerate}
	Then a sequence $\F_{\widetilde \varphi}$ does exist such that \eqref{13} holds. If, in particular, $\D=\H$, then two sequences $\F_{\widetilde \varphi}$ and $\F_{\widetilde \psi}$ do exist such that \eqref{14} holds.
\end{cor}

The previous corollary in the case $\lambda=1$ and $\D=\H$ is also connected to Theorem 5.3 of \cite{cfl} formulated for Hilbert spaces. Now, we show an application of the corollary.
%We remark also that in \cite[Theorem 2.1]{} the hypothesis (1) is not necessarily to conclude that the sequence $\{h_k\}_{k=1}^\infty$ is a frame for $\H$.  

{

\begin{exm}
	Let $\mathcal F_e$ be an orthonormal basis of $\H$. Let us define a new sequence $\mathcal{F}_\chi$ as follows $\chi_n=e_n+e_{n+1}$, $n\geq 1$ and let $\D$ be its linear span. As mentioned in \cite[Example 5.4.6]{Chris}, $\mathcal{F}_\chi$ is a Bessel sequence with Bessel bound $4$, but not a frame. However, it is has a bi-orthogonal sequence $\mathcal{F}_\xi$, where in particular $\displaystyle \xi_n=\sum_{k=1}^n (-1)^{n+k} e_k$. This implies that 
	$$
	f=\sum_{n=1}^\infty \pin{\xi_n}{f} \chi_n, \qquad \forall f\in \D. 
	$$ 
	Now let $\mathcal{F}_\varphi$ be the sequence given by $$\varphi_n=\sum_{k=1}^n (-1)^{n+k} \gamma_{n,k} e_k, \qquad n\geq 1,$$ where $\{\gamma_{n,k}\}_{n\geq k\geq 1}$ is some complex sequence such that, for some $\lambda \neq 0$, $\displaystyle \sum_{n\geq k} |\lambda\gamma_{n,k}-1|^2$ is convergent for every $k\geq 1$ and $\displaystyle \sup_{k\geq 1} \Big (\sum_{n\geq k} |\lambda\gamma_{n,k}-1|^2 \Big )<\infty $. Then $\F_{\varphi-\xi}$ is a Bessel sequence with Bessel bound $\displaystyle \sup_{k\geq 1} \Big (\sum_{n\geq k} |\lambda\gamma_{n,k}-1|^2 \Big )$. Indeed $$ \sum_{n=1}^\infty  |\pin{f}{\varphi_n-\xi_n}|^2=\sum_{k=1}^\infty \Big (\sum_{n\geq k}|\lambda\gamma_{n,k}-1|^2  \Big ) |\pin{f}{e_k}|^2\leq \sup_{k\geq 1} \Big (\sum_{n\geq k} |\lambda\gamma_{n,k}-1|^2 \Big ) \|f\|^2$$
	for every $f \in \H$. If, in particular, 
	\begin{equation}
		\label{cond1/4}
		\sup_{k\geq 1}\Big (\sum_{n\geq k} |\lambda\gamma_{n,k}-1|^2 \Big )< \frac 14
	\end{equation} by Corollary \ref{cor2.8}, point 2,  condition \eqref{13} holds for some sequence $\F_{\widetilde \varphi}$. We observe that the inequality \eqref{cond1/4} holds, for instance, taking $\gamma_{n,k}=\frac{1}{\lambda}(\frac{\epsilon}{n-k+1}+1)$ for any $\lambda\neq 0$ and $0\leq \epsilon<\frac{3}{2\pi^2}$.  
\end{exm}
}

Theorem \ref{theom41} above gives also a generalization of Lemma \ref{pertcn}.

\begin{cor}
	Let $\F_\varphi$ and $\F_{\widetilde \varphi}$ be sequences such that $\displaystyle f=\sum_{n=0}^\infty \pin{\varphi_n}{f}\widetilde \varphi_n$ for every $f\in \H$. Let $\F_{\psi}$ be a sequence, $0\leq\alpha<1$ and $\lambda\neq 0$ such that 
	\begin{equation}
	\label{pertDc_n}
	\left \|\sum_{n=0}^\infty c_n(\lambda\psi_n- \widetilde \varphi_n) \right\|\leq \alpha \left \| \sum_{n=0}^\infty c_n\widetilde\varphi_n \right \|,  
	\end{equation}
	for every finite complex sequence $\{c_n\}$. Then a sequence $\F_{\widetilde \psi}$ does exist such that $$
	f=\sum_{n=0}^\infty \pin{\widetilde\psi_n}{f}\psi_n, \qquad \forall f\in \Hil.
	$$
\end{cor}
\begin{proof}
	Firstly, by a limit argument, we note that \eqref{pertDc_n} is valid also for those complex sequences $\{c_n\}$ such that $\displaystyle \sum_{n=0}^\infty c_n\widetilde\varphi_n$ is convergent. Then, if we fix $f\in \H$ and choose $c_n=\pin{\varphi_n}{f}$ for any $n$, we obtain condition
	\eqref{pertDbasis} with $\D=\H$. Hence, the conclusion follows by Theorem \ref{theom41}.  
\end{proof}

\subsection{An example from Quantum Mechanics}

This section is based on results previously deduced in \cite{bagjmaa2022} and in \cite{bagspringer}, to which we refer for more details. Here we will only state few essential facts, useful for us here.

Let us first introduce the following functions:
\be
\eta_n(x)=\frac{1}{\sqrt{2^n\,n!\sqrt{\pi}}}H_n(x)e^{-s_A(x)}=e_n(x)e^{-\sqrt{2}\alpha x}, \quad \chi_n(x)=\frac{1}{\sqrt{2^n\,n!\sqrt{\pi}}}H_n(x)e^{-s_B(x)}=e_n(x)e^{\sqrt{2}\alpha x}.
\label{ex1}\en
Here $H_n$ is the $n$-th Hermite polynomial, $s_A(x)=\frac{x^2}{2}+\sqrt{2}\alpha x$, $s_B(x)=\frac{x^2}{2}-\sqrt{2}\alpha x$, $\alpha$ is a real parameter, and $n=0,1,2,3,\ldots$. Moreover,
$$
e_n(x)=\frac{1}{\sqrt{2^n\,n!\sqrt{\pi}}}H_n(x)e^{-\frac{x^2}{2}}
$$
is the $n$-th eigestate of the Hamiltonian of the quantum mechanical harmonic oscillator, $H_0=\frac{1}{2}\left(\hat p^2+\hat q^2\right)$, where $\hat q$ and $\hat p$ are the (self-adjoint) position and momentum operators, satisfying (in a suitable sense, for instance, on the dense set of Schwartz functions $\Sc(\mathbb{R})$) the following canonical commutation relation: $[\hat q,\hat p]=i\1$. We refer to \cite{mess}, or to any other monograph on Quantum Mechanics, for many details on the harmonic oscillator. 

First of all we observe that $\F_\eta=\{\eta_n, \,n\geq0\}$ and $\F_\chi=\{\chi_n, \,n\geq0\}$ are not Riesz bases for $\ltwo$, since they are obtained by multiplying the o.n. basis $\F_e=\{e_n, \,n\geq0\}$, with an unbounded operator, with unbounded inverse. However, it is easy to check that $(\F_\eta,\F_\chi)$ are $D(\mathbb{R})$-quasi bases. Here $D(\mathbb{R})$ is the set of $C^\infty$-functions with compact support, which is known to be dense in $\ldue$. Let us call $H=H_0-\sqrt{2}\,\alpha\,\frac{d}{dx}$. Then $H^\dagger=H_0+\sqrt{2}\,\alpha\,\frac{d}{dx}$, and we have
$$%\be
H\eta_n=E_n\eta_n, \qquad H^\dagger \chi_n=E_n\chi_n, \qquad E_n=n+\frac{1}{2}+\alpha^2=E_n^0+\alpha^2,
$$%\label{ex2}\en
where $E_n^0=n+\frac{1}{2}$ is the $n$-th eigenvalue of $H_0$, with eigenvector $e_n$. The operators $H$ and $H^\dagger$ are a particular case of the Hamiltonians of a {\em shifted harmonic oscillator}, \cite{bagspringer}, with a complex shift. This kind of operators are nowadays widely analysed in the context of PT-quantum mechanics and its relatives, see \cite{bagspringer,bagbookPT} and references therein. 

In the rest of this section we will show how to use the strategy proposed  in details in Example \ref{exa1}, in a slightly modified form, to deform $\F_\eta$ and $\F_\chi$, and to deduce some quantum-mechanical motivated operators connected to these deformations. In particular, our choice of functions in (\ref{ex1}) corresponds to the choice $\beta=-\alpha$  in Example \ref{exa1}. Hence we have 
$$%\be
Q=\1(1-\alpha)+2\alpha P_\sigma, \qquad Q^{-1}=\frac{1}{1-\alpha}\left(\1-\frac{2\alpha}{1+\alpha}\,P_\sigma\right),
$$%\label{ex3}\en
and
$$
\varphi_n=(\1+\alpha P_\sigma)\eta_n, \qquad \psi_n=(\1(1-\alpha)+\alpha P_\sigma)\chi_n,
$$
together with
$$
\widetilde{\varphi}_n=\left(\1-\frac{\alpha}{1+\alpha}\,P_\sigma\right)\chi_n, \qquad 
\widetilde{\psi}_n=\frac{1}{1-\alpha}\left(\1-\alpha P_\sigma\right)\eta_n.
$$
To refine further these formulas, we fix $\sigma(x)$ as follows: $\sigma(x)=e_0(x)=\frac{1}{\pi^{1/4}}e^{-x^2/2}$.  Then we can rewrite the action of $P_\sigma$ on $\eta_n$ and $\chi_n$ as follows:
$$%\be
P_\sigma\eta_n= \pi_n(\alpha)e_0, \qquad P_\sigma\chi_n= \pi_n(-\alpha)e_0,
$$%\label{ex4}\en
where
$$
\pi_n(\alpha)=\frac{1}{\sqrt{2^n\,n!\sqrt{\pi}}}\int_{\mathbb{R}}H_n(x)e^{-x^2-\sqrt{2}\alpha x}\,dx=\frac{(-\alpha)^n}{\sqrt{n!}}e^{\alpha^2/2},
$$
$n=0,1,2,3,\ldots$. Then the functions above become, for instance,
$$
\varphi_n(x)=\eta_n(x)-\frac{(-\alpha)^{n+1}}{\sqrt{n!}}e^{\alpha^2/2}\,e_0(x), \qquad \widetilde\varphi_n(x)=\chi_n(x)-\frac{(-\alpha)^{n+1}}{(1+\alpha)\sqrt{n!}}e^{\alpha^2/2}\,e_0(x),
$$
and so on. All these functions are simply the original ones, but shifted by some other function which is proportional to $e_0(x)=\sigma(x)$.  We see that these proportionality constants depend, in particular, on $n$. What is also interesting for possible applications to physics is that all these deformed functions are eigenstates, all with the same eigenvalues $E_n$, of different operators which one can easily find by suitable deformations of $H$ and $H^\dagger$. In particular we can introduce the following (for the moment, formal) operators:
$$%\be
H_1=(\1+\alpha P_\sigma)H(\1+\alpha P_\sigma)^{-1}, \qquad H_2=(\1-\alpha R_\sigma)H^\dagger(\1-\alpha R_\sigma)^{-1},
%\label{ex5}
$$%\en
as well as
$$
H_3=\left(\1-\frac{\alpha}{1+\alpha}\,P_\sigma\right)H^\dagger \left(\1-\frac{\alpha}{1+\alpha}\,P_\sigma\right)^{-1}, \qquad H_4=(\1-\alpha P_\sigma)H(\1-\alpha P_\sigma)^{-1}.
$$ 
We see that $H_4$ coincides with $H_1$, with $\alpha$ replaced by $-\alpha$. This is also true for $H_3$, which is equal to $H_2$ with the same replacement, as one can check with easy computations. It is clear that the various formulas above are all well defined on, some dense subset of $\ltwo$. In particular, for instance, $H_1$ is well defined on $\Lc_\varphi$, while $H_2$ is well defined on $\Lc_\psi$, the linear spans of $\F_\varphi$ and $\F_\psi$ respectively. The following eigenvalue equations are satisfied:
$$%\be
H_1\varphi_n=E_n\varphi_n, \quad H_2\psi_n=E_n\psi_n, \quad H_3\widetilde\varphi_n=E_n\widetilde\varphi_n,\quad H_4\widetilde\psi_n=E_n\widetilde\psi_n,
$$%\label{ex6}\en
$\forall n\geq0$. The Hamiltonians $H_1$ and $H_2$ (and $H_3$ and $H_4$ consequently) can be rewritten as follows:
$$%\be
H_1=H-\sqrt{2}\alpha^2\left(P_\sigma x+\frac{1}{1+\alpha}xP_\sigma\right), \qquad H_2=H^\dagger+\sqrt{2}\alpha^2\left(xP_\sigma +\frac{1}{1-\alpha}P_\sigma x\right),
$$%\label{ex7}\en
which contain not only the shifted versions of the Hamiltonian of the harmonic oscillator $H_0$, $H$ and $H^\dagger$, but also present extra terms which manifestly contribute to breaking down {\em even more} self-adjointness of the new Hamiltonians.
We end our analysis noticing that, as expected, when $\alpha=0$ all the operators, eigenfunctions and eigenvalues collapse to those of the standard harmonic oscillator.

\vspace{2mm}

{\bf Remark:--} It is maybe useful to notice here that our choice of the harmonic oscillator as physically motivated example is somehow natural: our analysis originates from the analysis of pseudo-bosons, \cite{bagspringer}, which are strongly connected to particular deformations of the harmonic oscillator. However, it is clear that the same strategy can be adapted to other systems, like, just to cite one, to the P\"osch-Teller potential
$$
V_{\lambda}(x)=\frac{\lambda(\lambda-1)}{\sin^2x}\, ,
$$
where $\lambda\geq1$ and $x\in (0,\pi)$. For $\lambda > 1$, this potential is a regularization of the infinite square well with center in $\frac{\pi}{2}$ and of length $\pi$  ($\lambda =1$) and extrapolates both the latter and the harmonic oscillator (for small $\vert x-\pi/2\vert$). The Hamiltonian of the particle is  $$H_{\lambda}=-\frac{d^2}{d x^2}+V_{\lambda}(x)\, , $$
and the eigenvalue equation for $H_{\lambda}$ can be explicitly solved, (see for instance \cite{antoine_etal_01} and references therein):
\be
H_{\lambda}\,e^{\lambda}_n(x)=\epsilon^{\lambda}_n e^{\lambda}_n(x)\, ,
\label{28}\en
where $\epsilon^{\lambda}_n=(n+\lambda)^2$ and
\be
e^{\lambda}_n(x)=K^{\lambda}_n\,\sin^{\lambda} x\,\mathrm{C}_n^{\lambda}\left(\cos x\right)\,.
\label{29}\en
Here
$$
K^{\lambda}_n=\Gamma(\lambda)\frac{2^{\lambda-1/2}}{\sqrt{\pi}}\sqrt{\frac{n!(n+\lambda)}{\Gamma(n+2\lambda)}}
$$
is a normalization constant and $\mathrm{C}_n^{\lambda}$ is the  Gegenbauer polynomial of degree $n$ \cite{magnus66}. The set $\{e^{\lambda}_n(x)\}$ is an orthonormal basis.

Now, as we have done in (\ref{ex1}), we can act on $e^{\lambda}_n(x)$ with an unbounded multiplication with unbounded inverse $S(x)$. For instance, we could consider $S(x)=\tan\left(x+\frac{\pi}{2}\right)$, and we can repeat the same construction described before.

\section{Convergence in weak sense}
\label{sec3}

In many physical examples the eigenstates of a certain non self-adjoint Hamiltonian do not allow us to expand all vectors of $\Hil$, or even of some dense $\D$, as discussed in Section \ref{sec2}. However, these vectors satisfy a weak version of the resolution of the identity as that in (\ref{dquasi}), see \cite{bagspringer} and references therein. This motivates the analysis proposed in this section. 

\begin{theo}
		\label{th_pertHquasi}
	Let $\F_\varphi$ and $\F_{\psi}$ be two sequences of vectors in $\Hil$ such that $\displaystyle \sum_{n=0}^\infty \pin{f}{\varphi_n}\pin{\psi_n}{g}$ exists for every $f,g\in \H$ and let $0\leq\alpha<1$, $\lambda\neq 0$ be such that 
	\begin{equation}
		\label{pertHquasi}
		\left |\lambda \sum_{n=0}^\infty  \pin{f}{\varphi_n}\pin{\psi_n}{g}-\pin{f}{g} \right |\leq \alpha \|f\|\|g\|, \qquad \forall f, g\in \H. 
	\end{equation}
	Then there exist $\F_{\widetilde \varphi}$ and $\F_{\widetilde \psi}$ such that
	\be\label{res_weak}
	\pin{f}{g}=\sum_{n=0}^\infty \pin{f}{\varphi_n}\pin{\widetilde\varphi_n}{g}=\sum_{n=0}^\infty \pin{f}{\psi_n}\pin{\widetilde\psi_n}{g}, \qquad \forall f, g\in \H.
	\en
	Moreover,  $\pin{\varphi_n}{\widetilde{\varphi}_m}=\pin{\widetilde\psi_n}{{\psi}_m}$  for every $n,m$.  % $\pin{\varphi_n}{\widetilde\varphi_m}=\delta_{n,m}$ if and only if $\pin{\psi_n}{\widetilde\psi_m}=\delta_{n,m}$.
\end{theo}

\begin{proof}
We start observing that, from (\ref{pertHquasi}), 
$$
	\left |\sum_{n=0}^\infty \pin{f}{\varphi_n}\pin{\psi_n}{g} \right |\leq \frac{\alpha+1}{\lambda}\|f\|\|g\|,\qquad \forall f, g\in \H.
$$
Therefore, the operator acting as  $Qf=\sum_{n=0}^\infty \pin{\varphi_n}{f}\psi_n$, $\forall f\in \H$, in weak sense\footnote{This means that $\pin{Qf}{g}=\sum_{n=0}^\infty \pin{f}{\varphi_n}\pin{\psi_n}{g}$, $\forall  f,g\in \H.$}, is bounded.  It is also easy to check that $\|\lambda Q-\1\|<1$. Therefore, by Lemma \ref{lem_pert_op}, $Q$ is invertible and $Q^{-1}\in \B(\Hil)$. Let us now define the vectors
$$
\widetilde{\varphi}_n=Q^{-1}\psi_n, \qquad \widetilde{\psi}_n=(Q^{-1})^\dagger\varphi_n=(Q^\dagger)^{-1}\varphi_n. 
$$
We have, using the definition of $Q$,
$$
\sum_{n=0}^\infty \pin{f}{\varphi_n}\pin{\widetilde\varphi_n}{g}=\sum_{n=0}^\infty \pin{f}{\varphi_n}\pin{\psi_n}{(Q^{-1})^\dagger g}=\pin{Qf}{(Q^{-1})^\dagger g}=\pin{f}{g}, \qquad \forall f,g\in\H.
$$
Similarly we have
$$
\sum_{n=0}^\infty \pin{f}{\psi_n}\pin{\widetilde\psi_n}{g}=\sum_{n=0}^\infty \pin{f}{\psi_n}\pin{\varphi_n}{Q^{-1} g}=\pin{Q^\dagger f}{Q^{-1} g}=\pin{f}{g}, \qquad \forall f,g\in\H.
$$
Finally, our last claim can be proved as in Theorem \ref{theom41}.
\end{proof}

{
Using Theorem \ref{th_pertHquasi} it is possible to construct, in particular, out of two $\Hil$-quasi bases, more $\Hil$-quasi bases. The underlying idea, shown in the following example, is the same of that used in Example \ref{exa1}. 

\begin{exm}\label{exa2}
Let $\F_\eta$ and $\F_{\chi}$ be two $\Hil$-quasi bases:
$$
\sum_{n=0}^\infty \pin{f}{\eta_n}\pin{\chi_n}{g}=\sum_{n=0}^\infty \pin{f}{\chi_n}\pin{\eta_n}{g}=\pin{f}{g}, \qquad \forall f, g\in\Hil.
$$
Let $X,Y\in \B(\Hil)$ be two non necessarily invertible operators with $\|X\|+\|Y\|+\|X\|\|Y\|<1$. Then,  $\F_\varphi=\{\varphi_n=(\1+X)\eta_n\}$ and $\F_\psi=\{\psi_n=(\1+Y)\chi_n\}$ satisfy (\ref{pertHquasi}) for $\lambda=1$, so there exist $\F_{\widetilde \varphi}$ and $\F_{\widetilde \psi}$ such that \eqref{res_weak} holds. As in Example \ref{exa1}, it is possible to take for instance $X=\alpha\,P_\sigma$ and $Y=\beta\,R_\sigma$, $\alpha,\beta>0$ with $\alpha+\beta+\alpha\beta<1$. Thus, we obtain again 
	$$
	\varphi_n=\eta_n+\alpha P_\sigma\eta_n, \qquad \psi_n=\chi_n+\beta R_\sigma\chi_n,
	$$
	and
	$$
	\widetilde{\varphi}_n=\left(\1-\frac{\alpha}{1+\alpha}\,P_\sigma\right)\chi_n, \qquad
	\widetilde{\psi}_n=\left(\1-\frac{\beta}{1+\beta}\,R_\sigma\right)\eta_n.
	$$
\end{exm}
}

\vspace{3mm}

Relaxing condition \eqref{pertHquasi} we can state the following. 

\begin{theo}
	\label{th_pertDquasi}
	Let $\D_1$ and $\D_2$ be two dense subspaces of $\H$ and let $\F_\varphi$ and $\F_{\psi}$ be two sequences of vectors in $\Hil$ such that $\displaystyle \sum_{n=0}^\infty \pin{f}{\varphi_n}\pin{\psi_n}{g}$ for any $f\in \D_1,g\in \D_2$ and let $0\leq\alpha<1$ and $\lambda \neq 0$ be such that 
	\begin{equation}
	\label{pertDquasi}
	\left |\lambda \sum_{n=0}^\infty \pin{f}{\varphi_n}\pin{\psi_n}{g}- \pin{f}{g} \right |\leq \alpha \|f\|\|g\|, \qquad \forall f\in \D_1, g\in \D_2. 
	\end{equation}
	Then two sequences $\F_{\widetilde \varphi}$, $\F_{\widetilde \psi}$ and two subspaces $\widetilde\D_1$, $\widetilde\D_2$ do exist such that
	\begin{equation}\label{pertDquasi1}
	\pin{f}{g}=\sum_{n=0}^\infty \pin{f}{\varphi_n}\pin{\widetilde\varphi_n}{g}, \quad \forall f\in \D_1,g \in \widetilde\D_2
	\end{equation}
	and
	\begin{equation}\label{pertDquasi2}
	\pin{f}{g}=\sum_{n=0}^\infty \pin{f}{\psi_n}\pin{\widetilde\psi_n}{g}, \quad \forall f\in \widetilde\D_1,g \in \D_2.
	\end{equation}
	Moreover,  $\pin{\varphi_n}{\widetilde{\varphi}_m}=\pin{\widetilde\psi_n}{{\psi}_m}$ for every $n,m$. 
\end{theo}

\begin{proof}
	We only sketch the proof being similar to that of Theorem \ref{th_pertHquasi}.  The operator $Q:\D_1\to \H$ defined by 
	$\displaystyle \pin{Qf}{g}=\sum_{n=0}^\infty \pin{f}{\varphi_n}\pin{\psi_n}{g}$, for $ f\in \D_1,g\in \D_2$, is bounded and admits a bounded and invertible extension $\overline{Q}:\H\to \H$ with $\overline{Q}^{-1}\in \B(\H)$. The sequences $\F_{\widetilde \varphi}$, $\F_{\widetilde \psi}$ and the subspaces $\widetilde\D_1$, $\widetilde\D_2$ are given by 
	\[
	\widetilde{\varphi}_n=\overline{Q}^{-1}\psi_n, \qquad \widetilde{\psi}_n=(\overline{Q}^{-1})^\dagger\varphi_n=({Q}^\dagger)^{-1}\varphi_n, \qquad \widetilde \D_1=\overline{Q} \D_1={Q} \D_1, \qquad \widetilde \D_2={Q}^\dagger \D_2. \qedhere
	\]
\end{proof}

Of course, Remarks \ref{rem2.8}.2 and \ref{rem2.8}.3  hold also in the context of Theorems \ref{th_pertHquasi} and \ref{th_pertDquasi}. We observe that an explicit characterization of $\widetilde \D_1$ and $\widetilde \D_2$ is not so easy and depends strongly on the explicit form of $Q$, which, in turns, depend clearly on $\F_\varphi$ and $\F_{\psi}$.

\begin{exm}\label{exl2R2}
	We come back to the sequences defined in Example \ref{exl2R1}, i.e. $\F_\varphi=\{\varphi_n\}$ and $\F_\psi=\{\psi_n\}$, 
	\be
	\varphi_n(x)=\rho_1(x)\,c_n(x)\qquad \psi_n(x)=\rho_2(x)\,c_n(x),
	\label{f2.2}\en
	where $c=\{c_n\}$ is an orthonormal basis of $\ltwo$, $\rho_j$, $j=1,2$, are two Lebesgue-measurable functions with real values, with $\n{\rho_1\rho_2}_\infty <\infty$ and $ \rho_1(x)\rho_2(x)>c>0$ in $\R$. However, here we do not assume that $\rho_1$ or $\rho_2$ belongs to $\mathcal{L}^\infty(\R)$ as we did in Example \ref{exl2R1}. Still, $\rho_1$ and $\rho_2$ are chosen to ensure that $\varphi_n,\psi_n\in \Lc^2(\R)$ for every $n$. Let us introduce the subspaces
	$$%\be
	%\V=\left\{f(x)\in\ltwo:\,\rho_j(x)f(x)\in\ltwo, \,j=1,2\right\}.\\
	\D_j=\left\{f\in\ltwo:\,\rho_jf\in\ltwo\right\}\qquad j=1,2.
	$$%\label{f1}\en
	First of all we observe that, as observed in  Example \ref{exl2R1}, they are both dense subspaces in $\Lc^2(\mathbb{R})$ and moreover if $\rho_j\in\Lc^\infty(\mathbb{R})$, then $\D_j=\ltwo$. \\	
	Let us take $f\in \D_1, g\in\D_2$. Then,  because of (\ref{f2.2}),
	$$
	\sum_{n=0}^\infty \pin{f}{\varphi_n}\pin{\psi_n}{g}=\pin{\rho_1 f}{ \rho_2 g}
	$$
	where the right-hand side is well defined. Indeed, our present assumptions ensure that $\rho_1f, \rho_2g\in \ltwo$. 
	In other words, the operator $Q$ in the proof of Theorem \ref{th_pertDquasi} is the multiplication operator by $\rho_1\rho_2$ on $\D_1$, while $\overline{Q}$ is the multiplication operator by $\rho_1\rho_2$ on $\H$. Now let $0<\lambda<\frac{2}{\n{\rho_1\rho_2}_\infty}$. 
	We have
	$$
	\left |\lambda\sum_{n=0}^\infty \pin{f}{\varphi_n}\pin{\psi_n}{g}-\pin{f}{g} \right |=\left|\pin{(\lambda\rho_1 \rho_2-1 )f}{g}\right|\leq \alpha\|f\|\|g\|, \qquad \forall f\in \D_1, g\in\D_2
	$$
	where $\alpha=\|\lambda\rho_1 \rho_2-1\|_\infty$ which is smaller than $1$ by our hypothesis on $\rho_1,\rho_2$. 
	Thus \eqref{pertDquasi} is satisfied, so \eqref{pertDquasi1} and \eqref{pertDquasi2} hold. In particular, it is easy to see that also in this case $\widetilde \varphi_n(x)=\rho_1^{-1}(x)c_n(x)$ and $\widetilde \psi_n(x)=\rho_2^{-1}(x)c_n(x)$. 
\end{exm}

%\begin{rem}
%	It is clear that, when (\ref{f3}) is satisfied, the previous examples applies trivially for $\lambda=1$, since $\|1-\rho_1\overline \rho_2\|_\infty=0${\color{blue}, and $(\F_\varphi,\F_\psi)$ turn out to be $\V$-quasi bases, \cite{bagspringer}, where $\V=...$}. 	
%\end{rem} 

\section{Working with distributions: a concrete example}
\label{sec4}

In this section we will show that the results given in Theorem \ref{th_pertDquasi} can be extended outside $\Hil$. This will be discussed with an explicit example, recently considered in \cite{baggarg2022,bagspringer}. For that, to keep this paper self-contained, we start giving here few essential steps of this example.

Let us consider the following operators defined on $\Hil=\Lc^2(\mathbb{R})$: $\hat xf(x)=xf(x)$, $(\hat Dg)(x)=g'(x)$, the derivative of $g$, for all $f\in \D(\hat x)=\{h\in\Lc^2(\mathbb{R}): xh(x)\in \Lc^2(\mathbb{R} \}$ and $g\in \D(\hat D)=\{h\in\Lc^2(\mathbb{R}): h'\in \Lc^2(\mathbb{R}) \}$. Here $h'$ is the ordinary derivative\footnote{Later in this chapter we will work with the {\em weak derivative} of certain distributions.} of the square-integrable, differentiable, function $h$. Of course,  $\Sc(\mathbb{R})\subset \D(\hat x)\cap \D(\hat D)$. The adjoints of $\hat x$ and $\hat D$ in $\Hil$ are well known: $\hat x^\dagger=\hat x$ and $\hat D^\dagger=-\hat D$. We have $[D,x]f(x)=f(x)$, for all $f\in\Sc(\mathbb{R})$. In \cite{baggarg2022,bagspringer}, using the general settings of pseudo-bosonic operators, \cite{bagspringer}, the following sequences have been constructed, using the pseudo-bosonic ladder operators $a=\hat D$ and $b=\hat x$:
$$
\F_{\Phi}=\left\{\Phi(x)=\frac{x^n}{\sqrt{n!}}, \, n\geq0\right\}, \qquad \F_{\eta}=\left\{\eta_n(x)=\frac{(-1)^n}{\sqrt{n!}}\,\delta^{(n)}(x), \, n\geq0\right\}.
$$
We refer  to  \cite{baggarg2022,bagspringer} for the details of their construction, and for many of their properties. Here we just observe, as it is clear, that $\Phi_n,\eta_n\notin\ldue$, for no value of $n$. Despite of this, however, a sesquilinear form extending the usual scalar product in $\ldue$ can be introduced, and the sequences $\F_{\Phi}$ and $\F_{\eta}$ are biorthonormal, with respect to this form.

To achieve this aim, we use an idea discussed  in \cite{vlad}. This is based on the simple fact that the scalar product between two {\em good } functions $f(x)$ and $g(x)$, for instance $f(x),g(x)\in\Sc(\mathbb{R})$, can be written in terms of a convolution between $\overline{f(x)}$ and $\tilde{g}(x)=g(-x)$: $\left<f,g\right>=(\overline{f}* \tilde{g})(0)$. Hence it is natural to define the scalar product between two elements $F(x), G(x)\in\Sc'(\mathbb{R})$ as the following convolution\footnote{We are using $F(x)$ and $G(x)$ to indicate the tempered distributions since we have in mind only those {\em regular distributions} for which this is possible. This is possible since we are {\bf not} claiming that all the elements of $\Sc'(\mathbb{R})$ can be multiplied.}:
\be
\left<F,G\right>=(\overline{F}* \tilde{G})(0),
\label{exdist1}\en
whenever this convolution exists, which is not always true. Notice that, in order to compute $\left<F,G\right>$, it is first necessary to compute $(\overline{F}* \tilde{G})[f]$, $f(x)\in\Sc(\mathbb{R})$, and this can be done by using the equality $(\overline{F}* \tilde{G})[f]=\left<F,G*f\right>$ which, again, is not always well defined. It is maybe useful to stress that $(\overline{F}* \tilde{G})[f]$ represents here the action of $(\overline{F}* \tilde{G})(x)$ on the function $f(x)$.

With this definition we  find that
\be\left<\Phi_n,\eta_m\right>=\delta_{n,m},
\label{exdist2}\en
as claimed before: $\F_{\Phi}$ and $\F_{\eta}$ are biorthonormal, with respect to the scalar product we are considering here.

\vspace{2mm}

Let us introduce the following set of functions:
\be
\D=\Lc^1(\mathbb{R})\cap \Lc^\infty(\mathbb{R})\cap A(\mathbb{R}),
\label{exdist3}
\en
where $A(\mathbb{R})$ is the set of  entire real analytic functions, which admit expansion in Taylor series, convergent everywhere in $\mathbb{C}$. Of course $\D$ contains many functions of $\Sc(\mathbb{R})$, but not all. For instance, the function 
$$
f(x)=\left\{
\begin{array}{ll}
	exp\left\{-\left(x^2+\frac{1}{x^2}\right)\right\}, \qquad x\neq0\\
	0, \hspace{3.7cm} x=0,\\
\end{array}
\right.
$$
is in $\scr$. But, since all its derivatives in $x=0$ are zero, $f(x)\notin A(\mathbb{R})$.
 Also, any function in $\D$ automatically belongs to $\ldue$. Indeed, taking $f\in\D$ we have, with standard notation\footnote{Notice that in this formula we are using $\|.\|_1$, $\|.\|_2$ and $\|.\|_\infty$ since they all appear. In the rest of the paper we have mostly simply used $\|.\|$ to indicate the norm in $\ldue$, since there is no possible misunderstanding.},
$$
\|f\|_2^2=\int_{\mathbb{R}}|f(x)|^2dx\leq \sup_{x\in\mathbb{R}}|f(x)|\int_{\mathbb{R}}|f(x)|dx=\|f\|_\infty\|f\|_1<\infty.
$$
 In \cite{baggarg2022,bagspringer} it has been proved that, for all $f,g\in\D$, $\left<f,\Phi_n\right>$ and $\left<\eta_n,g\right>$ are well defined, and that
\be
\left<f,g\right>=\sum_{n=0}^\infty\,\left<f,\Phi_n\right>\left<\eta_n,g\right>=\sum_{n=0}^\infty\,\left<f,\eta_n\right>\left<\Phi_n,g\right>,
\label{exdist4}\en
so that $(\F_{\Phi},\F_{\eta})$ are $\D$-quasi bases. It should be stressed that, since $\scr\not\subset\D$, it is not clear if $\D$ is dense or not in $\ldue$ but this is not very relevant in the present context, where the role of the Hilbert space is not so central.

Now, as in Example \ref{exa2}, we first deform $(\F_{\Phi},\F_{\eta})$, and then we use this (class of) deformations to discuss how the main results in Theorem \ref{th_pertDquasi} can be extended outside $\Hil$.

Let us consider two linear operators $M$ and $L$ on $\ldue$, with adjoints $M^\dagger$ and $L^\dagger$, satisfying the following features:

{\bf p1.} $L$ and $M$ leave $\D$ stable;

{\bf p2.} $M^\dagger \Phi_n$ and $L^\dagger\eta_n$ are well defined in $\Sc'(\mathbb{R})$, $\langle f,M^\dagger \Phi_n\rangle$ and $\langle f,L^\dagger \eta_n\rangle$ exist for all $f\in\D$ and $\forall n\geq0$, and the following equalities are satisfied:
\be
\langle Mf,\Phi_n\rangle=\langle f,M^\dagger \Phi_n\rangle, \qquad \langle Lf,\eta_n\rangle=\langle f,L^\dagger \eta_n\rangle.
\label{exdist5}\en
Notice that, while it is granted the existence of $\langle Mf,\Phi_n\rangle$ and $\langle Lf,\eta_n\rangle$, because of what proved in \cite{baggarg2022} and for the stability of $\D$ under $L$ and $M$, the existence of the right-hand sides of the equalities in (\ref{exdist5}) is not automatic. 

{\bf p3.} We have
\be\|M^\dagger L-\1\|<1
\label{exdist6}\en

Because of {\bf p2} the new vectors $\varphi_n=M^\dagger \Phi_n$ and $\psi_n=L^\dagger\eta_n$ are well defined in $\Sc'(\mathbb{R})$. Now, using (\ref{exdist4}), together with {\bf p2.} again, we can easily check that
$
\sum_{n=0}^\infty\,\left<f,\varphi_n\right>\left<\psi_n,g\right>=\langle f, M^\dagger Lg\rangle,
$
$\forall f,g\in\D$, so that
$$
\left|\sum_{n=0}^\infty\,\left<f,\varphi_n\right>\left<\psi_n,g\right>-\langle f, g\rangle\right|=\left|\langle f,(M^\dagger L-\1) g\rangle\right|\leq \|M^\dagger L-\1\|\|f\|\|g\|.
$$
This inequality, because of (\ref{exdist6}) is exactly the one in (\ref{pertDquasi}), with $\lambda=1$ and $\D_1=\D_2=\D$. The only (crucial!) difference is that the vectors $\varphi_n$ and $\psi_n$, except for quite peculiar choices of $M$ and $L$ (which, even if they exist, are not useful for us, here), are not in $\ldue$. Now the question we want to answer is if and how it is possible to construct the {\em dual sequences} of $(\F_{\varphi},\F_{\psi})$, i.e. two new sequences of vectors $\F_{\widetilde\varphi}=\{\widetilde\varphi_n\}$ and  $\F_{\widetilde\psi}=\{\widetilde\psi_n\}$ (belonging to some space to be identified) such that the analogous of (\ref{pertDquasi1}) and (\ref{pertDquasi2}) are satisfied.

This construction can be carried out easily if we slightly modify {\bf p1.} and {\bf p2.} as follows:

{\bf p1'.} $L$ and $M$, invertible, leave $\D$ stable together with $(L^{-1})^\dagger$ and $(M^{-1})^\dagger$;

{\bf p2'.} $M^\dagger \Phi_n$, $L^\dagger\eta_n$, $M^{-1} \eta_n$ and $L^{-1}\Phi_n$ are well defined in $\Sc'(\mathbb{R})$, $\langle f,M^\dagger \Phi_n\rangle$, $\langle f,L^\dagger \eta_n\rangle$, $\langle f,M^{-1} \eta_n\rangle$ and $\langle f,L^{-1} \Phi_n\rangle$ exist for all $f\in\D$ and $\forall n\geq0$, and the following equalities are satisfied:
\be
\langle Mf,\Phi_n\rangle=\langle f,M^\dagger \Phi_n\rangle, \qquad\langle (L^{-1})^\dagger f,\Phi_n\rangle=\langle f,L^{-1}\Phi_n\rangle, 
\label{exdist7}\en
as well as
\be
 \langle Lf,\eta_n\rangle=\langle f,L^\dagger \eta_n\rangle, \qquad \langle (M^{-1})^\dagger f,\eta_n\rangle=\langle f,M^{-1} \eta_n\rangle.
\label{exdist8}\en
Of course, our previous comments on the existence of these quantities can be repeated.

Under {\bf p1'.} and {\bf p2'.} it is easy to identify $\widetilde\varphi_n$ and  $\widetilde\psi_n$. Indeed, if we define
\be
\widetilde\varphi_n=M^{-1}\eta_n, \qquad \widetilde\psi_n=L^{-1}\Phi_n,
\label{exdist9}
\en
we see first that these are well defined vectors of $\Sc'(\mathbb{R})$. Secondly, taking $f,g\in\D$ and using (\ref{exdist8}) and (\ref{exdist9}), together with $\varphi_n=M^\dagger \Phi_n$, we have
$$
\sum_{n=0}^\infty\,\left<f,\varphi_n\right>\left<\widetilde\varphi_n,g\right>=\sum_{n=0}^\infty\,\left<Mf,\Phi_n\right>\left<\eta_n,(M^{-1})^\dagger g\right>=\left<Mf,(M^{-1})^\dagger g\right>=\left<f, g\right>.
$$
Here we have also used (\ref{exdist4}), and the fact that both $Mf$ and $(M^{-1})^\dagger g$ belong to $\D$, because of {\bf p1'.}. Similarly we can check that, again for $f,g\in\D$, $\sum_{n=0}^\infty\,\left<f,\psi_n\right>\left<\widetilde\psi_n,g\right>=\left<f,g\right>$. Then we have recovered similar results as those in Theorem \ref{th_pertDquasi}, but in a different function space.

\subsection{A simple class of choices of $M$ and $L$}

Though the above hypotheses on $M$ and $L$ can appear rather strong, it is not hard to construct examples which work nicely, in our context. The easiest choice is possibly the following:
\be
Mf(x)=m(x)f(x), \qquad Lf(x)=l(x)f(x),
\label{exex1}\en
for all $f\in\ldue$. Here $m,l$ are real, invertible, functions such that $m,l,m^{-1},l^{-1}\in\Lc^\infty(\mathbb{R})$, and they all in $A(\mathbb{R})$. It is clear that $M$ and $L$ are bounded, with bounded inverse, and self-adjoint: $M=M^\dagger$ and $L=L^\dagger$. It is also easy to see that {\bf p1'.} is satisfied. This is because, for instance, if $f\in\D$, then $Mf(x)=m(x)f(x)$ also belongs to $\D$: $mf\in\Lc^\infty(\mathbb{R})\cap A(\mathbb{R})$, being the product of two such functions. Moreover, since
$$
\|Mf\|_1=\int_{\mathbb{R}}|m(x)f(x)|\,dx\leq \|m\|_\infty\|f\|_1<\infty,
$$
$Mf\in\Lc^1(\mathbb{R})$, as we had to check. Similar estimates hold for the other operators in {\bf p1'.}. 

As for {\bf p2'.}, first we notice that $m \Phi_n$, $l\eta_n$, $m^{-1} \eta_n$ and $l^{-1}\Phi_n$ are all well defined in $\Sc'(\mathbb{R})$, because of the properties of $m$ and $l$, and of the definition of $\Phi_n$ and $\eta_n$. The second part of  {\bf p2'.}, i.e. the fact that $\langle f,M^\dagger \Phi_n\rangle$, $\langle f,L^\dagger \eta_n\rangle$, $\langle f,M^{-1} \eta_n\rangle$ and $\langle f,L^{-1} \Phi_n\rangle$ all exist, can be deduced from what was shown in \cite{bagspringer}, where the convergence of a series involving this kind of scalar products is proven. The equalities in (\ref{exdist7}) and (\ref{exdist8}) are trivial.
 
Finally, $m$ and $l$ must still satisfy the bound in (\ref{exdist6}): any such pair can be used to deform the original $\D$-quasi basis $(\F_{\Phi},\F_{\eta})$, to find two new $\D$-quasi bases $(\F_{\varphi},\F_{\widetilde\varphi})$ and $(\F_{\psi},\F_{\widetilde\psi})$ out of them. If we take, for instance, $l(x)=\frac{1}{2m(x)}$, it is easy to check that $\|M^\dagger L-\1\|\leq\|ml-1\|_\infty=\frac{1}{2}$, so that (\ref{exdist6}) is satisfied for all choices of $m$. Of course, other choices of $m$ and $l$, and more generally other form of operators $M$ and $L$ could be found. This is part of our future projects.

\section{Conclusions}\label{sec6}

We have discussed how $\D$-quasi bases and more general sequences can be perturbed giving rise to other sequences which share with the original ones the possibility of obtaining closure relations of different kind, in $\Hil$ or even outside $\Hil$. In particular, we have shown how sequences which are not, for instance, $\D$-quasi bases, under some natural conditions allow us to define other sequences which are $\D$-quasi bases, when considered in pairs with the original ones.

Some examples are discussed. In particular, one quantum mechanically example connected to the shifted harmonic oscillator has been explored in some details.

Among our plans for the future, we would like to explore the role of (pseudo-)bosonic ladder operators in our framework, also in view of the possible implications in the analysis of manifestly non self-adjoint Hamiltonians.  Also, we would like to apply the perturbation results to more generic multipliers \cite{corso_mult1,corso_mult2}; in fact, Lemma \ref{lem_pert_op} has been applied in \cite{corso_mult2} to study an aspect of the spectra of dual frame multipliers.

In Section \ref{sec4} we have given some ideas to work with distributions, showing in this way that the role of Hilbert spaces is not really essential, sometimes. In this perspective, another plan is to extend the perturbation result of the first section in the setting of triplets $\D[t] \subset \H \subset \D^\times[t^\times]$, as considered in \cite{ttt_19},  and chains of Hilbert spaces, as considered in \cite{act}.  The idea of investigating perturbation results in the distributional context is supported by the recent investigation of bases and frames in this context \cite{ttt_19,cor_ts}.

\section*{Acknowledgements}

F.B.  acknowledges partial financial support from Palermo University (via FFR2021 "Bagarello") and from G.N.F.M. of the INdAM. R. C. acknowledges partial financial support from Palermo University (via FFR2021 "Corso") and from G.N.A.M.P.A. of the INdAM. This work has been done within the activities of the Gruppo UMI Teoria dell’Approssimazione e Applicazioni.

\end{document}